\documentclass[reqno,10pt]{amsproc}
\usepackage{amssymb}
\usepackage{amsmath}
\usepackage{amsthm}
\usepackage{amsfonts}
\usepackage{microtype}
\usepackage{xcolor}
\usepackage{geometry}
\usepackage{mathrsfs}
\usepackage{mathtools}
\usepackage{thmtools}

\usepackage[backend=biber,style=ieee,sorting=none,backref=true,citestyle=numeric-comp,giveninits=false]{biblatex}
\usepackage{tikz,tikz-cd}
\usetikzlibrary{arrows,cd,decorations.markings,arrows.meta}

\tikzset{->-/.style={decoration={markings,mark=at position .5 with {\arrow{>}}},postaction={decorate}}}
\tikzset{-<-/.style={decoration={markings,mark=at position .5 with {\arrow{<}}},postaction={decorate}}}
\tikzstyle{box}=[fill=white, draw=black, shape=rectangle, inner sep=14pt]
\tikzstyle{forward arrow}=[->-]
\tikzstyle{backward arrow}=[-<-]

\definecolor{myurlcolor}{rgb}{0,0,0.4}
\definecolor{mycitecolor}{rgb}{0,0.5,0}
\definecolor{myrefcolor}{rgb}{0.5,0,0}
\usepackage[breaklinks=true]{hyperref}
\hypersetup{colorlinks,
	linkcolor=myrefcolor,
	citecolor=mycitecolor,
	urlcolor=myurlcolor,
hypertexnames=false}
\usepackage[capitalize]{cleveref}

\theoremstyle{plain}
\declaretheorem[name=Dummy,numberwithin=section]{dummy}
\declaretheorem[name=Theorem,sibling=dummy]{thm}
\declaretheorem[name=Lemma,sibling=dummy]{lem}
\declaretheorem[name=Proposition,sibling=dummy]{prop}

\declaretheorem[name=Definition,sibling=dummy]{defn}

\theoremstyle{remark}
\declaretheorem[name=Example,sibling=dummy]{ex}
\declaretheorem[name=Remark,sibling=dummy]{rem}

\numberwithin{equation}{section}
\Crefname{equation}{}{}		

\allowdisplaybreaks

\let\originalleft\left
\let\originalright\right
\renewcommand{\left}{\mathopen{}\mathclose\bgroup\originalleft}
\renewcommand{\right}{\aftergroup\egroup\originalright}

\usepackage{enumitem}
\setlist[enumerate]{label=(\roman*),itemsep=5pt,topsep=8pt}
\setlist[itemize]{label=$\triangleright$,itemsep=5pt,topsep=6pt}
\Crefformat{enumi}{#2#1#3}

\newcommand{\newterm}[1]{\textbf{#1}}

\newcommand{\beq}{\begin{equation}}
\newcommand{\eeq}{\end{equation}}
\newcommand{\N}{\mathbb{N}}

\newcommand{\C}{\mathbb{C}}

\newcommand{\Calg}{\mathcal{A}}		
\newcommand{\mfrak}{\mathfrak{m}}		
\newcommand{\Op}{\operatorname{Op}}		
\newcommand{\ev}{\operatorname{ev}}		
\newcommand{\Cont}{\mathsf{Cont}}		
\newcommand{\eps}{\varepsilon}

\newcommand{\id}{\mathsf{id}}		
\newcommand{\st}{\mathsf{st}}		

\newcommand{\FinSet}{\mathsf{FinSet}}
\newcommand{\Set}{\mathsf{Set}}

\newcommand{\Meas}{\mathsf{Meas}}
\newcommand{\inst}{\mathcal{E}}		
\newcommand{\instb}{\mathcal{F}}	
\newcommand{\instmon}{\mathcal{I}}	

\usepackage{hyphenat}

\usepackage{todonotes}

\begin{document}



\title{The quantum instrument monad}

\author{Tobias Fritz}

\address{Department of Mathematics, University of Innsbruck, Austria}
\email{tobias.fritz@uibk.ac.at}

\thanks{\textit{Acknowledgements.} We thank Sam Staton for extensive discussion and pointers to the literature as well as Tom{\'a}{\v{s}} Gonda and Antonio Lorenzin for further discussion.}

\begin{abstract}
	Monads are a ubiquitous structure in functional programming used for modelling computational effects.
	For example, the state monad models the effect of a computation interacting with a memory system.
	Here we introduce the quantum instrument monad $\instmon_\Calg$, which models the effect of a computation interacting with a quantum system with algebra of observables $\Calg$.
	It can be thought of as a noncommutative generalization of the state monad.
	We construct this quantum instrument monad in two versions: a finitary version on the category of sets and a measure-theoretic version on the category of measurable spaces
	(the latter under the assumption that $\Calg$ is a type I von Neumann algebra with separable predual).
	Both versions are strong monads.
	The construction of the measure-theoretic version is based on a new notion of integral of a quantum-operation-valued function against a state-valued measure.
\end{abstract}

\maketitle
\tableofcontents

\section{Introduction}

In operational terms, an interaction with a quantum system generally has two aspects: it can produce classical data, and it can change the state of the system.
The standard formalism that captures these two aspects simultaneously is that of a \newterm{quantum instrument}~\cite{DaviesLewis,Ozawa84,Heinosaari_Ziman}.
From this point of view, observables, state-update rules and quantum channels arise as special cases of instruments.
Thus if one wants to model computations that may interact with a quantum system, then doing so in terms of quantum instruments is a natural choice.

Second, since Moggi's seminal work~\cite{Moggi} it is well-understood that the category-theoretic notion of monad provides a general framework for how a computation can interact with its environment.
The classical state monad is the basic example: it models programs that may read and update a memory state while returning an output.
The main idea of the present paper is that interactions with a quantum system can be modelled by a monad as well, namely the \newterm{quantum instrument monad}, which is the quantum analogue of the classical state monad.
The quantum instrument monad models the effect of performing an arbitrary instrument on a fixed quantum system and returning the instrument outcome as output.

This perspective is particularly natural in quantum programming, where one repeatedly combines classical control with quantum data.
There is a substantial literature on quantum programming languages and their semantics; see for example the surveys~\cite{Selinger_survey,Gay_survey,Valiron_QPL_survey}.
Representative functional and semantic approaches include Selinger's language QPL~\cite{Selinger_QPL}, the language QML of Altenkirch and Grattage~\cite{Altenkirch_Grattage}, the quantum lambda calculus of Selinger and Valiron~\cite{Selinger_Valiron}, and the quantum IO monad~\cite{Quantum_IO_Monad}.
At a more implementation-oriented level, languages such as Quipper~\cite{Quipper} and QWIRE~\cite{QWIRE} likewise emphasize the disciplined combination of classical host-language control with quantum operations and measurement.
What seems to be missing, however, is a monadic construction whose effect is specifically ``performing an arbitrary instrument on a fixed quantum system''.
Such a construction is useful for the semantics of functional programming interfaces in which a computation can manipulate or measure an ambient quantum resource or interact with it in any other way permitted by quantum theory.

In this paper, we construct two versions of the quantum instrument monad in the Schr\"odinger picture.
The first one is a finitary version on the category of sets, and the second is a measure-theoretic version on the category of measurable spaces.
In both cases, the objects of the category are outcome spaces, and the morphisms are post-processings of outcomes.\footnote{The use of such a category mirrors the \emph{generalized measurement theories} of~\cite{GMTs}. The main difference is that that setting is only concerned with measurements and not with instruments, i.e.~post-measurement states are not considered.}
Conceptually, the monad works the same in both cases: for a fixed quantum system modelled by an algebra of observables $\Calg$, the functor $\instmon_\Calg$ takes an outcome space $X$ as input and returns $\instmon_\Calg X$, the set of all instruments with outcomes in $X$.
The monad has a unit with components
\begin{equation}
	\label{eta}
	\eta \: : \: X \longrightarrow \instmon_\Calg X,
\end{equation}
which takes every outcome $x \in X$ to the trivial instrument which always returns $x$ and does nothing to the system.
The multiplication has components
\begin{equation}
	\label{mu}
	\mu \: : \: \instmon_\Calg\instmon_\Calg X \longrightarrow \instmon_\Calg X,
\end{equation}
which encodes the sequential composition of instruments: given an instrument $\instb$ which produces outcomes \emph{in the space of instruments} $\instmon_\Calg X$, we can perform $\instb$ to obtain $\inst \in \instmon_\Calg X$ as outcome, and then perform $\inst$ to obtain a final outcome in $X$.

\subsection*{Strong monads}

A monad $T$ on a cartesian monoidal category is \newterm{strong} if it is equipped with a natural way of pairing a pure value with an effectful computation,
\[
	\st \: : \: X\times T Y \longrightarrow T(X\times Y),
\]
in a way that is compatible with the monad structure.
In functional programming terms, this is what lets an effectful computation be used in a context containing ordinary variables: the pure variable is passed along while the effectful computation is performed.
This is part of the standard categorical semantics of computational effects in the sense of Moggi~\cite{Moggi}, and thus one generally expects monads in functional programming to be strong.

For monads on $\Set$, it is well-known that there is a unique strength,\footnote{See e.g.~\cite[Example~3.7 and Proposition~4.7]{McDermottUustalu}.} given by
\begin{equation}
	\label{strength_formula}
	\st(x,\inst) \coloneqq T(y\mapsto(x,y))(\inst).
\end{equation}
In particular, the classical state monad and the finitary quantum instrument monad do not require a separate construction of a strength.
But for the measure-theoretic monad on $\Meas$, using~\eqref{strength_formula} does not automatically work, because the resulting map $X\times T Y\to T(X\times Y)$ must be shown to be measurable.
We therefore construct the strength explicitly in \Cref{prop:ozawa_strength}.

\subsection*{Related literature}

Most importantly, Booth, Leichtle, Rice and Worrall have developed a quantum instrument monad in the Heisenberg picture in independent and concurrent work. The technical details differ.
In particular, the construction of the quantum instrument integral is different. Nevertheless, the resulting monads appear to be closely related.
We refer the to~\cite{BLRW} for this complementary perspective.

Our construction is also related in spirit, but not equivalent, to other uses of monads in quantum foundations and quantum computation, such as the quantum monad on relational structures of Abramsky, Soares Barbosa, de Silva and Zapata~\cite{ASSZ}, and the quantum monadology of Sati and Schreiber~\cite{SS}.
Staton's work on algebraic effects with linear parameters~\cite{StatonAlgebraicEffectsQuantum} constructs a strong monad on $[\mathsf{Bij}, \Set]$ as a semantics for a quantum programming language operating on qubits, where the monad is generated by operations corresponding to preparing a fresh qubit, unitary evolution and measurement.
The main difference is that our monads are built directly from the instrument formalism itself, with the monadic variable representing the classical outcome space; in this sense they are the direct generalization of the classical state monad to the quantum setting.
A closely related construction is given by the linear state monads investigated by M{\o}gelberg, Rennela and Staton, on which we comment more in \cref{rem:linear_state_monads}.

\subsection*{Notation and conventions}

We write $\Set$ for the category of sets and functions, and $\Meas$ for the category of measurable spaces and measurable maps.
Throughout, $\Calg$ denotes a fixed von Neumann algebra of observables with predual $\Calg_*$, whose positive normalized elements are the states of the quantum system under consideration.
In natural transformations like~\eqref{eta} and~\eqref{mu}, we usually omit the component index in order to avoid clutter, which amounts to an abuse of notation where the same symbol denotes a natural transformation and every individual component.

\section{The classical state monad}

In this section, we review the classical state monad~\cite{Wadler_monads} in the language of instruments.
This will set the stage for the quantum instrument monad as a noncommutative generalization.
In order to make the analogy as close as possible, our presentation of the state monad will be somewhat nonstandard.

Let $S$ be a fixed set, thought of as the set of states of a classical system.
We assume that this set comes equipped with a fixed element $0 \in S$, to be thought of as an unnormalized state analogous to the zero density operator in the quantum case.
Then by analogy with the concept of quantum operation, we could say that a \emph{deterministic operation} is any map $S \to S$ with $0 \mapsto 0$, so that
\[
	\Op_S \coloneqq \{ \Phi : S\to S \mid \Phi(0)=0\}.
\]
Such an operation is \emph{normalized} if no other state maps to $0$.
If we have a bunch of operations indexed by outcomes in such a way that ignoring the outcome results in a normalized operation, then we have a deterministic instrument:

\begin{defn}
	\label{defn:classical_instrument}
	A \newterm{deterministic instrument} with values in a set $X$ is a map
	\begin{align*}
		\inst \: : \: X \times S & \longrightarrow S \\
		(x, s) & \longmapsto \inst(x, s)
	\end{align*}
	satisfying the following properties:
	\begin{enumerate}
		\item\label{item:classical_instrument_zero}
			For every $x \in X$, the map $\inst(x, -) : S \to S$ is such that $\inst(x, 0) = 0$.
		\item\label{item:classical_instrument_map}
			For every $s \in S\setminus\{0\}$, there is exactly one $x \in X$ with $\inst(x, s) \neq 0$.
	\end{enumerate}
\end{defn}

The idea is that the unique $x$ with $\inst(x, s) \neq 0$ is the outcome of performing the instrument $\inst$ on the state $s$, and the post-measurement state is $\inst(x, s)$.
For every set $X$, define
\[
	\instmon_S X
	\coloneqq
	\{ \text{deterministic instruments with values in } X \}.
\]
This becomes a functor $\instmon_S:\Set\to\Set$ by post-processing outcomes: for a map $f:X\to Y$ and an instrument $\inst\in\instmon_S X$, define $f_*\inst\in\instmon_S Y$ by
\[
	(f_*\inst)(y,s)
	\coloneqq
	\sum_{x\in f^{-1}(y)} \inst(x,s).
\]
Here, we are using a partially defined addition on $S$, where a sum is defined if and only if at most one summand is nonzero, and in that case the sum is equal to that summand.
Thus the definition of $f_*$ amounts to saying that if the outcome of the instrument $\inst$ on state $s$ is $x$, then the outcome of $f_*\inst$ on $s$ is $f(x)$, with the same post-measurement state $\inst(x, s)$.
The post-processed $f_*\inst$ is again a deterministic instrument: for every nonzero $s$, the unique nonzero outcome of $f_*\inst$ is obtained by applying $f$ to the unique nonzero outcome of $\inst$.

We turn to the definition of the monad structure.
The unit $\eta$ takes every outcome $x\in X$ to the trivial instrument which always returns $x$ and does nothing to the state:
\begin{align*}
	\eta \: : \: X &\longrightarrow \instmon_S X \\
	x &\longmapsto \left((x',s)\mapsto \delta_{x,x'}\,s\right).
\end{align*}
The monad multiplication is sequential composition of instruments.
For $\instb\in\instmon_S\instmon_S X$, define $\mu(\instb)\in\instmon_S X$ by
\begin{equation}
	\label{mu_state}
	\mu(\instb)(x,s)
	\coloneqq
	\sum_{\inst\in\instmon_S X} \inst(x,\instb(\inst,s)).
\end{equation}
This sum is defined in the same partial sense as above.
This is an instrument again because, for $s=0$, every summand is zero, while for $s\neq 0$ there is a unique $\inst_s\in\instmon_S X$ with $\instb(\inst_s,s)\neq 0$, and then a unique $x\in X$ with $\inst_s(x,\instb(\inst_s,s))\neq 0$.
For this unique $\inst_s$, the formula reduces to
\begin{equation}
	\label{mu_state_reduced}
	\mu(\instb)(x,s)
	=
	\inst_s(x,\instb(\inst_s,s)).
\end{equation}
Operationally, this formula encodes the following procedure: given an instrument $\instb$ which produces outcomes in the space of instruments $\instmon_S X$, we perform $\instb$ on the state $s$ to obtain an inner instrument $\inst_s$ as outcome and an updated state $\instb(\inst_s,s)$; then we perform that inner instrument $\inst_s$ on the updated state to get a final outcome in $X$ and the final state~\eqref{mu_state_reduced}.

\begin{thm}
	\label{thm:classical_state_monad}
	With these definitions, $\instmon_S$ is a monad on $\Set$.
\end{thm}

Although it would be possible to verify the monad axioms directly, this is somewhat obscure especially for the associativity law.
Instead, we conduct the proof of the monad laws by constructing an embedding into a continuation monad (\Cref{app:continuation_monads}).

\begin{proof}
	Functoriality follows by the partial sum computation
	\[
		(g_*f_*\inst)(z,s)
		=
		\sum_{y\in g^{-1}(z)}\sum_{x\in f^{-1}(y)}\inst(x,s)
		=
		\sum_{x\in (g\circ f)^{-1}(z)}\inst(x,s).
	\]
	We prove the monad axioms by embedding $\instmon_S$ into a continuation monad.
	For every set $X$, define
	\begin{align*}
		\theta \: : \: \instmon_S X &\longrightarrow \Cont_{\Op_S} X \\
		\inst &\longmapsto \left(L\mapsto \left[s\mapsto \sum_{x\in X} L(x) \left(\inst(x,s)\right) \right]\right),
	\end{align*}
	where $L:X\to \Op_S$ is a continuation, and we leave it understood that in every term, $L$ gets applied to $\inst(x,s)$.
	The sum is defined in the same partial sense as above.
	The resulting map $s\mapsto \sum_x L(x) \left(\inst(x,s)\right)$ is indeed an operation, since every summand is zero when $s=0$.

	We now prove the following claims about $\theta$, writing $\instmon$ for $\instmon_S$ and $\Cont$ for $\Cont_{\Op_S}$.
	\begin{enumerate}
		\item $\theta$ is natural.

			For $f:X\to Y$, this means that the square
			\[
				\begin{tikzcd}
					\instmon X \arrow[r, "\theta"] \arrow[d, "f_*"'] & \Cont\, X \arrow[d, "\Cont\, f"] \\
					\instmon Y \arrow[r, "\theta"] & \Cont\, Y
				\end{tikzcd}
			\]
			commutes.
			We start with $\inst\in\instmon X$ top left.
			Taking the lower left path, applying the result to a continuation $L:Y\to \Op_S$ and evaluating on a state $s\in S$ gives
			\begin{align*}
				\theta(f_*\inst)(L)(s)
			&=
			\sum_{y\in Y}L(y) \left((f_*\inst)(y,s)\right) \\
			&=
			\sum_{y\in Y}L(y) \left(\sum_{x\in f^{-1}(y)}\inst(x,s)\right) \\
			&=
			\sum_{y\in Y} \sum_{x \in f^{-1}(y)} L(y) \left(\inst(x,s)\right) \\
			&=
			\sum_{x\in X}L(f(x)) \left(\inst(x,s)\right) \\
			&=
			\theta(\inst)(L\circ f)(s)
			=
			(\Cont\, f)(\theta(\inst))(L)(s).
			\end{align*}
		\item $\theta$ is injective.

			Indeed, for $x\in X$, consider the continuation $L_x:X\to \Op_S$ given by
			\[
				L_x(x') \coloneqq \delta_{x,x'}\id_S,
			\]
			by which we mean the constant zero operation for $x'\neq x$ and the identity operation for $x'=x$.
			Then
			\[
				\theta(\inst)(L_x)=\inst(x,-).
			\]
			This shows that $\theta(\inst)$ determines $\inst(x,-)$ for every $x\in X$, and hence $\inst$ itself, making $\theta$ injective.

		\item $\theta$ preserves the unit in the sense that the diagram
			\[
				\begin{tikzcd}
			& X \arrow[dl, "\eta"'] \arrow[dr, "\eps"] & \\
				\instmon X \arrow[rr, "\theta"'] && \Cont\, X
				\end{tikzcd}
			\]
			commutes: again evaluating on $L:X\to \Op_S$ and $s\in S$ gives
			\[
				\theta(\eta(x))(L)(s)
				=
				\sum_{x'\in X}L(x') \left(\delta_{x,x'}s\right)
				=
				L(x) \left(s\right)
				=
				\eps(x)(L)(s).
			\]
		\item $\theta$ also preserves the multiplication in the sense that the diagram
			\[
				\begin{tikzcd}[column sep=large]
					\instmon\instmon X \arrow[r, "\theta"] \arrow[d, "\mu"'] & \Cont\,\instmon X \arrow[r, "\Cont\,\theta"] & \Cont\, \Cont\, X \arrow[d, "\nu"] \\
					\instmon X \arrow[rr, "\theta"'] && \Cont\, X
				\end{tikzcd}
			\]
			commutes.
			Indeed for $\instb\in\instmon\instmon X$, evaluating on $L:X\to \Op_S$ and $s\in S$ gives the following.
			Here $\ev_L:\Cont\,X\to\Op_S$ is evaluation at $L$, so that $(\ev_L \circ \theta)(\inst)=\theta(\inst)(L)$.
			\begin{align*}
				(\theta\circ\mu)(\instb)(L)(s)
				&=
				\theta(\mu(\instb))(L)(s) \\
				&=
				\sum_{x\in X}L(x) \left(\mu(\instb)(x,s)\right) \\
				&=
				\sum_{x\in X}L(x) \left(\sum_{\inst\in\instmon X}\inst(x,\instb(\inst,s))\right) \\
				&=
				\sum_{\inst\in\instmon X}
				\sum_{x\in X}L(x) \left(\inst(x,\instb(\inst,s))\right).
			\end{align*}
			On the other hand,
				\begin{align*}
					(\nu \circ (\Cont\,\theta) \circ \theta)(\instb)(L)(s)
					&=
					(\Cont\,\theta)(\theta(\instb))(\ev_L)(s) \\
					&=
					\theta(\instb)(\ev_L\circ\,\theta)(s) \\
					&=
					\sum_{\inst\in\instmon X}(\ev_L\circ\,\theta)(\inst) \left(\instb(\inst,s)\right) \\
					&=
					\sum_{\inst\in\instmon X}\theta(\inst)(L) \left(\instb(\inst,s)\right) \\
					&=
					\sum_{\inst\in\instmon X}
					\sum_{x\in X}L(x) \left(\inst(x,\instb(\inst,s))\right).
				\end{align*}
	\end{enumerate}
	Hence $\theta$ is natural, injective, and compatible with unit and multiplication.
	Therefore the naturality of $\eta$ and $\mu$ and the monad laws for $\instmon$ follow from the corresponding facts for $\Cont$, and $\theta$ is a morphism of monads.
\end{proof}

\section{The finitary quantum instrument monad}

We now turn to the quantum case, using the Schr\"odinger picture.
In this section, the quantum system is described by an algebra of observables $\Calg$ which we assume to be a von Neumann algebra, and states are normal linear functionals on $\Calg$, i.e.~elements of the predual $\Calg_*$ which are positive and normalized with respect to the canonical trace functional
\[
	\tau \: : \: \Calg_* \longrightarrow \C,
\]
which is given by the pairing with $1\in\Calg$.
In this section and the next, we will construct monads $\instmon_\Calg$ which model the possible instruments that can be performed on such a quantum system.\footnote{Modulo the caveat that the measure-theoretic version will require additional assumptions on $\Calg$, which force it to be of type I (\cref{rem:radon_nikodym_property}).}

In this section, we start this endeavour with a simpler version of the quantum instrument monad that lives on $\Set$ and is finitary.
Here, the word \emph{finitary} indicates both that each instrument has only finitely many possible outcomes, and that the resulting monad $\instmon_\Calg$ is a finitary monad, i.e.~its underlying endofunctor preserves filtered colimits.\footnote{This is equivalent to the statement that every element of $\instmon_\Calg X$ lies in the image of $\instmon_\Calg Y \to \instmon_\Calg X$ for some finite subset $Y \subseteq X$~\cite[Remark~2.3(2)]{AMMU}, and the latter is exactly what we formally mean by each instrument having only finitely many possible outcomes.}
The construction of this monad is going to be closely parallel to the classical state monad from the previous section.
We first recall a basic definition.

\begin{defn}
	A \newterm{quantum operation} on $\Calg$ is a completely positive linear map $\Phi : \Calg_* \to \Calg_*$ that is trace-nonincreasing, meaning that
	\[
		\tau(\Phi(\rho)) \leq \tau(\rho) \qquad \forall \rho \in \Calg_{*,+}.
	\]
	We write $\Op(\Calg)$ for the set of quantum operations.
\end{defn}

An instrument now consists of a family of quantum operations indexed by outcomes, such that the total operation obtained by ignoring the outcome is normalized.

\begin{defn}
	\label{defn:finitary_quantum_instrument}
	A \newterm{finitary quantum instrument} with values in a set $X$ is a map
	\begin{align*}
		\inst \: : \: X \times \Calg_* & \longrightarrow \Calg_* \\
		(x,\rho) & \longmapsto \inst(x,\rho)
	\end{align*}
	satisfying the following properties:
	\begin{enumerate}
		\item\label{item:quantum_instrument_cp}
			For every $x \in X$, the map $\inst(x, -) : \Calg_* \to \Calg_*$ is linear, completely positive, and trace-nonincreasing, meaning that
			\[
				\tau(\inst(x,\rho))\leq\tau(\rho)
				\qquad
				\forall \rho\in\Calg_{*,+}.
			\]
		\item\label{item:quantum_instrument_tp}
			The total map $\sum_{x \in X}\inst(x,-)$ is trace-preserving, i.e.
			\begin{equation}
				\label{eq:quantum_instrument_tp}
				\tau\left(\sum_{x\in X}\inst(x,\rho)\right)
				=
				\tau(\rho)
				\qquad
				\forall \rho\in\Calg_*.
			\end{equation}
		\item\label{item:quantum_instrument_finitary}
			We have $\inst(x, -) = 0$ for all but finitely many $x \in X$.
	\end{enumerate}
\end{defn}

Note that conditions~\ref{item:quantum_instrument_cp} and~\ref{item:quantum_instrument_tp} are analogous to conditions \cref{item:classical_instrument_zero} and~\cref{item:classical_instrument_map} in the definition of deterministic instrument (\cref{defn:classical_instrument}).
The finiteness condition~\ref{item:quantum_instrument_finitary} guarantees in particular that the sum in~\eqref{eq:quantum_instrument_tp} is effectively finite.

\begin{rem}
	For finite $X$, this is the same data as a completely positive trace-preserving map
	\begin{equation}
		\label{instrument_map}
		\Calg_* \longrightarrow \C^X \otimes \Calg_*,
	\end{equation}
	representing the Schr\"odinger-picture evolution from the initial system state to a classical pointer joint with the final system state.
	Formally, this description is obtained from $\inst$ by sending $\rho$ to
	\[
		\sum_{x \in X} e_x \otimes \inst(x,\rho),
	\]
	where $\{e_x\}_{x \in X}$ is the standard basis of $\C^X$.
\end{rem}

To obtain the monad, define
\[
	\instmon_\Calg X \coloneqq \{ \text{finitary quantum instruments with values in } X \}.
\]
This becomes a functor $\instmon_\Calg \: \colon \: \Set \to \Set$ with the action on morphisms given by post-processing of outcomes: for any $f : X \to Y$ and any $\inst \in \instmon_\Calg X$, we define $f_* \inst \in \instmon_\Calg Y$ by
\[
	(f_* \inst)(y,\rho) \coloneqq \sum_{x \in f^{-1}(y)} \inst(x,\rho),
\]
and it is straightforward to see that this is an instrument again.
Let us now turn to the monad structure.
The unit $\eta$ takes every outcome $x \in X$ to the trivial instrument which always returns $x$ and does nothing to the system,
which reads
\begin{align*}
	\eta \: : \: X &\longrightarrow \instmon_\Calg X \\
	x &\longmapsto \left( (x',\rho) \mapsto \delta_{x,x'}\,\rho \right).
\end{align*}
The multiplication $\mu$ is given by
\begin{align}
	\label{eq:mu}
	\begin{split}
		\mu \: : \: \instmon_\Calg\instmon_\Calg X &\longrightarrow \instmon_\Calg X \\
		\instb &\longmapsto \left( (x,\rho) \mapsto \sum_{\inst \in \instmon_\Calg X} \inst(x,\instb(\inst,\rho)) \right).
	\end{split}
\end{align}
In line with the same idea as for the classical state monad, the right-hand side is the Schr\"odinger-picture version of performing the outer instrument $\instb$ to obtain an inner instrument $\inst$ as outcome and an updated state $\instb(\inst,\rho)$, and then performing that inner instrument on the updated state to get an outcome $x$.
But let us convince ourselves that this actually makes sense formally.

\begin{lem}
	The right-hand side of~\eqref{eq:mu} indeed defines a finitary quantum instrument.
\end{lem}

\begin{proof}
	The sum converges in the trivial sense that only finitely many terms can be nonzero.
	For fixed $x \in X$, the map
	\[
		\rho \longmapsto \sum_{\inst \in \instmon_\Calg X} \inst(x,\instb(\inst,\rho))
	\]
	is linear and completely positive, being a finite sum of composites of linear completely positive maps.
	Finitarity holds because only finitely many $\inst$ satisfy $\instb(\inst, -) \neq 0$, and each such $\inst$ itself has finite support, so that only finitely many $x$ can satisfy $\mu(\instb)(x, -) \neq 0$.
	For $\rho\in\Calg_*$, the total trace is
	\begin{align*}
		\tau\left(\sum_{x \in X} \mu(\instb)(x,\rho)\right)
			&= \sum_{x \in X} \sum_{\inst \in \instmon_\Calg X}
			\tau\bigl(\inst(x,\instb(\inst,\rho))\bigr) \\
			&= \sum_{\inst \in \instmon_\Calg X}
			\tau\left(\sum_{x \in X}\inst(x,\instb(\inst,\rho))\right) \\
			&= \sum_{\inst \in \instmon_\Calg X} \tau(\instb(\inst,\rho)) \\
			&= \tau(\rho).
	\end{align*}
	By linearity, the total map is trace-preserving on all of $\Calg_*$.
	Since every summand is positive on positive inputs, this identity also implies that each component $\mu(\instb)(x,-)$ is trace-nonincreasing.
	Therefore $\mu(\instb)$ is indeed a finitary quantum instrument.
\end{proof}

\begin{thm}
	With these definitions, $\instmon_\Calg$ becomes a monad on $\Set$.
\end{thm}

\begin{proof}
	Post-processing preserves finitary quantum instruments, and functoriality follows by rearranging finite sums exactly as in the classical case.
	The unit is visibly a finitary quantum instrument, and the multiplication is well-defined by the preceding lemma.

	Put $R\coloneqq\Op(\Calg)$ and consider the continuation monad $\Cont_R$ on $\Set$.
	For the rest of the proof, we leave the subscripts $\Calg$ and $\Op(\Calg)$ implicit, writing $\instmon$ and $\Cont$ for the two monads.
	For every set $X$, define
	\begin{align*}
		\theta \: : \: \instmon X &\longrightarrow \Cont\, X \\
		\inst &\longmapsto \left(L\mapsto\left[\rho\mapsto\sum_{x\in X} L(x) \left(\inst(x,\rho)\right)\right]\right),
	\end{align*}
	where $L:X\to R$.
	The value $\theta(\inst)(L)$ is a quantum operation: it is a finite sum of composites of quantum operations, and for $\rho\in\Calg_{*,+}$,
	we have
	\[
		\tau\left(\sum_{x\in X}L(x) \left(\inst(x,\rho)\right)\right)
		\leq
		\sum_{x\in X}\tau(\inst(x,\rho))
		=
		\tau(\rho).
	\]

	The verification that $\theta$ is natural and compatible with the unit and multiplication is the finite-linear version of the proof of \Cref{thm:classical_state_monad}.
	The only difference is that the partially defined sums there are replaced by ordinary finite sums here, and the relevant operations are quantum operations rather than deterministic operations.
	Thus for $f:X\to Y$ and $L:Y\to R$,
	\begin{align*}
		\theta(f_*\inst)(L)(\rho)
		&=
		\sum_{y\in Y}L(y) \left(\sum_{x\in f^{-1}(y)}\inst(x,\rho)\right) \\
		&=
		\sum_{x\in X}L(f(x)) \left(\inst(x,\rho)\right)
		=
		(\Cont\, f)(\theta(\inst))(L)(\rho),
	\end{align*}
	where the middle equality uses linearity of the quantum operations $L(y)$.
	Injectivity is proved as before: for $x\in X$, take $L_x:X\to R$ defined as
	\[
		L_x(x') \coloneqq \delta_{x,x'}\id_{\Calg_*},
	\]
	and again $\theta(\inst)(L_x)=\inst(x,-)$.
	The rest of the proof is verbatim the same as before.
\end{proof}

\begin{rem}
	We expect that an analogous instrument monad can be defined for any generalized probabilistic theory, and possibly also for any generalized measurement theory in the sense of~\cite{GMTs}.
\end{rem}

We thank Sam Staton for pointing out the following alternative construction of the monad, and for suggesting a generalization which allows for quantum instruments between two distinct systems.

\begin{rem}
	\label{rem:linear_state_monads}
	\begin{enumerate}
		\item
			We now describe how our finitary quantum instrument monad is related to the \emph{linear state monads} of M{\o}gelberg and Staton~\cite{MogelbergStatonLinearUsage}.

			Let $\mathsf{C}$ be a closed symmetric monoidal category, thought of as classical data processing,
			and let $\mathsf{Q}$ be a $\mathsf{C}$-enriched category, thought of as quantum systems and processes,
			which has copowers by objects of $\mathsf{C}$.
			This means that, for every object $A \in \mathsf{Q}$, there is a $\mathsf{C}$-functor $(-)\otimes A : \mathsf{C}\to\mathsf{Q}$ satisfying the adjointness relation
			\[
				\mathsf{Q}(X\otimes A, B) \,\cong\, \mathsf{C}(X, \mathsf{Q}(A,B)).
			\]
			The counit is the evaluation morphism $\mathsf{Q}(A,B)\otimes A \to B$ corresponding to the identity on $\mathsf{Q}(A,B)$.
			For a fixed object $A \in \mathsf{Q}$, this induces a monad on $\mathsf{C}$ with underlying endofunctor
			\[
				X \longmapsto \mathsf{Q}(A,X\otimes A).
			\]
			The multiplication
			\[
				\mathsf{Q}(A,\mathsf{Q}(A,X\otimes A)\otimes A)
				\longrightarrow
				\mathsf{Q}(A,X\otimes A)
			\]
			is obtained by applying $\mathsf{Q}(A,-)$ to the counit.

			This monad behaves like an instrument monad:
			using element language for $\mathsf{C}$,
			an element of the left-hand side is a morphism $A\to\mathsf{Q}(A,X\otimes A)\otimes A$,
			and on an input state in $A$ it outputs a process $A\to X\otimes A$ together with an updated state of type $A$.
			The multiplication then applies the produced process to that updated state, and views the result as a single process $A\to X\otimes A$.

			If one takes $\mathsf{C}=\Set$ and lets $\mathsf{Q}$ be the Schr\"odinger-picture category whose objects are von Neumann algebras and whose hom-set $\mathsf{Q}(\Calg, \mathcal{B})$ is the set of completely positive trace-preserving maps $\Calg_*\to\mathcal{B}_*$,
			then copowers by arbitrary sets exist:
			the copower of $\Calg$ by a set $X$ is the von Neumann algebra
			\[
				\ell^\infty(X) \otimes \Calg
				\,\cong\,
				\ell^\infty(X,\Calg)
			\]
			with predual $\ell^1(X,\Calg_*)$.
			Hence $\mathsf{Q}(\Calg,X\otimes\Calg)$ is the set of completely positive trace-preserving maps $\Calg_*\to\ell^1(X,\Calg_*)$.
			For finite $X$, these are precisely the instruments in the form~\eqref{instrument_map},
			while for infinite $X$ they correspond to instruments with discrete outcomes in $X$ without a finite support condition.
			In this way, the linear state monad is another version of the quantum instrument monad.
			
			To obtain the finitary quantum instrument monad itself, one can use the equivalence between finitary monads on $\Set$
			and relative monads of type $\FinSet \to \Set$~\cite[p.~30]{AltenkirchChapmanUustalu}.
			It is thus enough to construct a relative monad from a relative version of the above adjunction,
			where copowers are only required to exist with respect to a full subcategory $\mathsf{C}_0 \subseteq \mathsf{C}$,
			which can be taken to be $\FinSet$.
			Then Rennela and Staton have also considered a relative version of linear state monads in their work on categorical models of the EWire calculus~\cite[Definition~2.1(6)]{RennelaStatonEWire}.
			Applying this to the same $\mathsf{C}$ and $\mathsf{Q}$ as above together with $\mathsf{C}_0 = \FinSet$ should recover our finitary quantum instrument monad.

			On the other hand, it does not seem to be possible to obtain our measure-theoretic quantum instrument monad (\cref{sec:ozawa_monad}) in a similar vein.
		\item
			Our quantum instrument monads assume that the quantum system before and after the instrument is the same.
			However, there is a more general framework for interactions with systems where the system type may change during the interaction,
			namely the \emph{parameterized monads} of Atkey~\cite{AtkeyParameterizedMonads}.

			Roughly, a parameterized monad is a generalization of a monad where the ``instruments'' have an initial and a final parameter.
			Thus, instead of a single endofunctor $T$, one has endofunctors $T_{A,B}$ functorially depending on parameter objects $A$ and $B$, and equipped with additional units $X\to T_{A,A}X$ and multiplications $T_{A,B}T_{B,C}X\to T_{A,C}X$ satisfying appropriate associativity and unitality conditions.
			M{\o}gelberg and Staton have shown that their linear state monads actually generate parameterized monads in this sense~\cite[Proposition~9.1]{MogelbergStatonLinearUsage}.

			In the present setting, the parameters should again be the category of von Neumann algebras and completely positive trace-preserving maps between their preduals.
			Like this, we expect that one can define parameterized generalizations of both the finitary and the measure-theoretic quantum instrument monads.
	\end{enumerate}
\end{rem}

\section{The measure-theoretic quantum instrument monad}
\label{sec:ozawa_monad}

We now describe a measure-theoretic version of the quantum instrument monad.
As in the finitary case, the quantum system is described by a von Neumann algebra $\Calg$, and we work in the Schr\"odinger picture on its predual $\Calg_*$, regarded as the Banach space of normal linear functionals on $\Calg$.
We now assume that $\Calg_*$ is separable and continue to write
\[
	\tau \: : \: \Calg_* \longrightarrow \C
\]
for the canonical trace functional, i.e.~the pairing with $1 \in \Calg$.
As outcome spaces, we now consider measurable spaces $(X, \Sigma_X)$, where $\Sigma_X$ is a $\sigma$-algebra of subsets of $X$.
By standard abuse of notation, we often leave $\Sigma_X$ implicit and just write $X$ for the measurable space.

In addition, for the integration argument below we assume that the predual $\Calg_*$ has the Radon--Nikodym property in the Banach space sense.

\begin{rem}
	\label{rem:radon_nikodym_property}
	This Radon--Nikodym property holds in particular in the type I situation $\Calg=\mathcal B(\mathcal H)$ with $\mathcal H$ a separable Hilbert space,
	since then $\Calg_*$ is the space of trace class operators $\mathcal T(\mathcal H)$.
	More generally, for any von Neumann algebra $\Calg$ the Radon--Nikodym property for $\Calg_*$ is equivalent to pure atomicity of $\Calg$~\cite{ChuScatteredRNP}.
	Thus, if $\Calg$ has separable predual, then $\Calg_*$ has the Radon--Nikodym property precisely when there is a direct sum decomposition
	\[
		\Calg \cong \bigoplus_{n=1}^{\infty} \mathcal B(\mathcal H_n)
	\]
	for a sequence of separable Hilbert spaces $\mathcal H_n$.
	In this case $\Calg_*$ is the $\ell^1$-direct sum of the corresponding spaces $\mathcal T(\mathcal H_n)$.
	In particular, the Radon--Nikodym property fails for diffuse examples such as $L^\infty([0,1])$ and for the type III algebras of algebraic quantum field theory.
	Finding a suitable quantum instrument monad that applies in such cases is an interesting open problem.
\end{rem}

Every $\Phi \in \Op(\Calg)$ is a contraction on $\Calg_*$.
Concerning the measurable structure, we equip $\Calg_*$ with the Borel $\sigma$-algebra of its norm topology, and $\Op(\Calg)$ with the smallest $\sigma$-algebra which makes the evaluation maps
\begin{align*}
	\Op(\Calg) &\longrightarrow \Calg_* \\
	\Phi &\longmapsto \Phi(\rho)
\end{align*}
for $\rho \in \Calg_*$ measurable.

\begin{rem}
	\label{rem:operation_space_standard_borel}
	Since $\Calg_*$ is a separable Banach space, its norm topology is Polish, and hence its Borel measurable structure is standard Borel.
	The measurable space $\Op(\Calg)$ is standard Borel as well.

	To see the latter, choose a countable dense subset $(\rho_n)_{n\in\N}$ of $\Calg_*$ and consider the map
	\begin{align*}
		\Op(\Calg) & \longrightarrow (\Calg_*)^\N \\
		\Phi & \longmapsto (\Phi(\rho_n))_{n\in\N}.
	\end{align*}
	This map is injective because quantum operations are contractions and hence determined by their values on a dense subset.
	Moreover, the $\sigma$-algebra above is already generated by these countably many evaluation maps, since evaluation at any $\rho\in\Calg_*$ is the pointwise limit of evaluations at a sequence of $\rho_n$ converging to $\rho$.
	The image is a closed subset of the Polish space $(\Calg_*)^\N$: a pointwise limit on the dense set of quantum operations extends uniquely to a contraction on $\Calg_*$, and complete positivity and subnormalization are preserved under pointwise norm limits.
	Thus $\Op(\Calg)$ is measurably isomorphic to a Borel subset of a Polish space, and is therefore standard Borel.
\end{rem}

\begin{lem}
	\label{lem:operation_evaluation_measurable}
	The pairing map
	\begin{align*}
		\Op(\Calg) \times \Calg_* &\longrightarrow \Calg_* \\
		(\Phi,\rho) &\longmapsto \Phi(\rho)
	\end{align*}
	is measurable as well.
\end{lem}

\begin{proof}
	This is an instance of the Carathéodory joint measurability theorem~\cite[Lemma~4.51]{AliprantisBorder}.
	Indeed, for fixed $\rho\in\Calg_*$, the map $\Phi\mapsto\Phi(\rho)$ is measurable by definition of the measurable structure on $\Op(\Calg)$.
	For fixed $\Phi\in\Op(\Calg)$, the map $\rho\mapsto\Phi(\rho)$ is continuous because $\Phi$ is a contraction.
	Since $\Calg_*$ is separable metrizable, the joint measurability theorem applies.
\end{proof}

The following is Ozawa's notion of instrument~\cite{Ozawa84}, except that we allow an arbitrary von Neumann algebra $\Calg$ with separable predual rather than only $\mathcal B(\mathcal H)$.

\begin{defn}
	\label{defn:ozawa_instrument}
	Let $X$ be a measurable space.
	Then a \newterm{quantum instrument} with values in $X$ is a map
	\[
		\inst : \Sigma_X \times \Calg_* \longrightarrow \Calg_*
	\]
	satisfying the following properties:
	\begin{enumerate}
		\item For every $S\in\Sigma_X$, the map $\inst(S,-):\Calg_*\to \Calg_*$ is a quantum operation.
		\item $\sigma$-additivity: for every pairwise disjoint family $(S_n)_{n \in \N}$ in $\Sigma_X$ and every $\rho \in \Calg_*$,
			\[
				\inst\left( \bigcup_{n \in \N} S_n,\rho \right)
				=
				\sum_{n \in \N} \inst(S_n,\rho)
			\]
			with convergence in the norm of $\Calg_*$.
		\item $\inst(X,-)$ is a channel, that is $\tau(\inst(X,\rho)) = \tau(\rho)$ for all $\rho \in \Calg_*$.
	\end{enumerate}
\end{defn}

\begin{ex}
	\begin{enumerate}
		\item Discrete projection-valued measures give canonical examples of quantum instruments via the standard projection postulate.
			Let $X$ be countable, and let $P_x\in\Calg$ be pairwise orthogonal projections with $\sum_{x\in X}P_x=1$ in the strong operator topology.
			The corresponding PVM is $P(S)\coloneqq\sum_{x\in S}P_x$ for all $S\subseteq X$.
			Under the projection postulate, the associated instrument is
			\[
				\inst(S,\rho)(a)
				\coloneqq
				\sum_{x\in S}\rho(P_xaP_x)
				\qquad
				\forall a \in \Calg.
			\]
			Thus, conditional on observing $x$, the unnormalized state update is $\rho\mapsto P_x\rho P_x$, in the usual Schr\"odinger-picture notation.
			For a coarse-grained event $S$, the instrument sums these updates over $x\in S$, which is exactly what makes $S\mapsto\inst(S,\rho)$ countably additive.
		\item Every POVM with outcomes in a measurable space $X$ becomes a quantum instrument after choosing post-measurement states, provided that the latter depend measurably on the measurement outcome.

			Concretely, let $F:\Sigma_X\to\Calg$ be a POVM, and let $x\mapsto\sigma_x$ be a measurable family of normal states on $\Calg$, meaning that $x\mapsto\sigma_x(a)$ is measurable for every $a\in\Calg$.
			Then the corresponding measure-and-prepare instrument is
			\[
				\inst(S,\rho)(a)
				\coloneqq
				\int_{x\in S} \sigma_x(a)\,\rho(F(\mathrm{d}x)),
				\qquad
				a\in\Calg,
			\]
			where $\rho(F(\mathrm{d}x))$ denotes integration with respect to the scalar measure $S\mapsto\rho(F(S))$.
			This instrument indeed has outcome statistics given by $F$, since
			\[
				\tau(\inst(S,\rho))
				=
				\rho(F(S)) \qquad \forall \rho \in \Calg_*.
			\]
			See also~\cite{CyconHellwig1977}, where these kinds of instruments have been called \emph{nuclear}.
	\end{enumerate}
\end{ex}

We now define
\[
	\instmon_\Calg X
	\coloneqq
	\{ \text{quantum instruments with outcomes in } X \}.
\]
We equip this set with the smallest $\sigma$-algebra which makes the evaluation maps
\begin{align*}
	\instmon_\Calg X &\longrightarrow \Op(\Calg) \\
	\inst &\longmapsto \inst(S,-)
\end{align*}
measurable for all $S \in \Sigma_X$.
By definition of the $\sigma$-algebra on $\Op(\Calg)$, this is equivalently the smallest $\sigma$-algebra which makes the evaluation maps in the form
\begin{align*}
	\instmon_\Calg X &\longrightarrow \Calg_* \\
	\inst &\longmapsto \inst(S,\rho)
\end{align*}
measurable for all $S \in \Sigma_X$ and $\rho \in \Calg_*$.

\begin{rem}
	If $X$ is standard Borel, then $\instmon_\Calg X$ is standard Borel as well, by the following arguments.

	If $X$ is countable, then $\Sigma_X=2^X$.
	Assuming $X = \N$ without loss of generality, an instrument is equivalently a sequence $(\Phi_n)_{n\in\N}$ of operations such that $\sum_n \Phi_n$ converges pointwise in norm and is a channel.
	After restricting to a countable dense subset of $\Calg_*$,
	these are countably many Borel conditions.
	Thus $\instmon_\Calg X$ is a Borel subset of the standard Borel space $\Op(\Calg)^\N$, which makes it standard Borel as well.

	By Kuratowski's theorem, the remaining case to consider is $X=2^\N$.
	Let $\mathcal R$ be the countable Boolean algebra of clopen cylinder sets in $2^\N$.
	Then the $\sigma$-algebra on $\instmon_\Calg(2^\N)$ is generated by the countably many evaluation maps $\inst\mapsto\inst(R,-)$ with $R\in\mathcal R$.
	Indeed, for fixed $\rho\in\Calg_*$, the class of all Borel sets $S\subseteq 2^\N$ for which $\inst\mapsto\inst(S,\rho)$ is measurable with respect to these evaluations is a Dynkin system containing $\mathcal R$, so that the claim follows by the $\pi$-$\lambda$ theorem.
	Hence the map
	\begin{align*}
		\instmon_\Calg(2^\N) &\longrightarrow \Op(\Calg)^{\mathcal R} \\
		\inst &\longmapsto \bigl(\inst(R,-)\bigr)_{R \in \mathcal R}
	\end{align*}
	is injective, again by the same Dynkin-system argument.
	Its image is described by countably many Borel conditions: finite additivity on the countable Boolean algebra $\mathcal R$ and the normalization condition that $\Phi_{2^\N}$ is a channel.
	Finite additivity is enough, since if a countable disjoint union of clopen subsets of $2^\N$ is again clopen, then compactness implies that the union is already finite.
	Thus, after evaluating at any $\rho\in\Calg_*$ and pairing with any $a \in (\Calg_*)^* = \Calg$, the usual Carathéodory extension theorem for premeasures on algebras gives the unique scalar extension from $\mathcal R$ to the Borel sets of $2^\N$.
	Since the continuous linear functionals on $\Calg_*$ separate points, these scalar extensions determine the required vector-valued extension.
	Therefore $\instmon_\Calg(2^\N)$ is a Borel subset of the standard Borel space $\Op(\Calg)^\mathcal R$.

	Finally, every uncountable standard Borel space is Borel isomorphic to $2^\N$ by Kuratowski's theorem~\cite[Theorem~15.6]{KechrisCDST}, while the countable case was handled above.
	Transporting instruments along such a Borel isomorphism gives a measurable isomorphism of the corresponding instrument spaces, so $\instmon_\Calg X$ is standard Borel.
\end{rem}

\begin{lem}
	\label{lem:ozawa_functor}
	The assignment $X \mapsto \instmon_\Calg X$ extends to a functor $\instmon_\Calg : \Meas \to \Meas$.
\end{lem}

\begin{proof}
	Let $f : X \to Y$ be measurable and let $\inst \in \instmon_\Calg X$.
	Define $f_*\inst \in \instmon_\Calg Y$ by
	\[
		(f_*\inst)(T,\rho) \coloneqq \inst(f^{-1}(T),\rho),
		\qquad T \in \Sigma_Y.
	\]
	The defining properties of a quantum instrument are inherited immediately from those of $\inst$.
	The induced map
	\begin{align*}
		\instmon_\Calg(f) : \instmon_\Calg X &\longrightarrow \instmon_\Calg Y \\
		\inst &\longmapsto f_*\inst
	\end{align*}
	is measurable because, for $T \in \Sigma_Y$, the map
	\[
		\inst \longmapsto 
		(f_*\inst)(T,-)
		=
		\inst(f^{-1}(T),-)
	\]
	is one of the generating evaluation maps on $\instmon_\Calg X$.
	Functoriality is immediate from functoriality of inverse images.
\end{proof}

We now define the unit and multiplication of the quantum instrument monad $\instmon_\Calg$.
The unit $\eta$ has components
\[
	\eta : X \longrightarrow \instmon_\Calg X
\]
sending $x \in X$ to the trivial instrument
\[
	\eta(x)(S,\rho) \coloneqq \delta_{x \in S}\,\rho.
\]
This instrument deterministically returns the outcome $x$ and does nothing to the quantum system.

For the multiplication, we use the \newterm{quantum operation integral} developed in \Cref{app:quantum_operation_integral} in order to define
\begin{equation}
	\label{eq:ozawa_multiplication}
\begin{aligned}
	\mu \: : \: \instmon_\Calg\instmon_\Calg X &\longrightarrow \instmon_\Calg X \\
	\instb & \longmapsto \left( (S, \rho) \mapsto \int_{\inst \in \instmon_\Calg X} \inst(S,-) \left(\instb(\mathrm{d}\inst,\rho)\right) \right).
\end{aligned}
\end{equation}
As before, an element $\instb$ is an instrument whose outcome is itself an instrument $\inst$ on $X$; the composite instrument first produces $\inst$ and then performs $\inst$.
More explicitly, for any input state $\rho$, we obtain a vector measure $\instb(-,\rho)$ on the outcome space $\instmon_\Calg X$ with values in $\Calg_*$, which describes the post-measurement state of the outer instrument $\instb$.
In order to calculate the quantum operation corresponding to the composite instrument conditioned on an outcome in $S$, we integrate the quantum operation $\inst(S,-)$ of the inner instrument $\inst$ against this vector measure.
Since this vector measure is positive and has total mass $\tau(\rho)$, the relevant finite variation assumption holds by \cref{rem:quantum_operation_integral_discussion}\ref{item:quantum_operation_integral_discussion_positive}.

\begin{prop}
	\label{prop:ozawa_unit_multiplication}
	The above maps $\eta$ and $\mu$ are measurable, and $\mu(\instb)$ is a quantum instrument for every $\instb \in \instmon_\Calg\instmon_\Calg X$.
\end{prop}

\begin{proof}
	We first prove measurability of $\eta$.
	For every $S \in \Sigma_X$ and $\rho \in \Calg_*$, the composite
	\[
		\begin{tikzcd}[column sep=3.2pc]
			X \arrow[r, "\eta"] &
			\instmon_\Calg X \arrow[r, "{\inst \mapsto \inst(S,\rho)}"] &
			\Calg_*
		\end{tikzcd}
	\]
	is the map
	\[
		x \longmapsto \eta(x)(S,\rho)=\delta_{x \in S}\rho,
	\]
	which is measurable.
	Thus $\eta$ is measurable by definition of the $\sigma$-algebra on $\instmon_\Calg X$.

	We next show that every $\mu(\instb)$ is a quantum instrument.
	For fixed $S$, the assignment $\rho \mapsto \mu(\instb)(S,\rho)$ is linear by linearity of the vector measure $\rho \mapsto \instb(-,\rho)$ and linearity of the integral in the measure variable.
	Complete positivity of $\mu(\instb)(S,-)$ is checked at matrix levels: if $[\rho_{ij}] \in M_n(\Calg_*)_+$, then
	\[
		R \longmapsto \instb(R,-)^{(n)}([\rho_{ij}])
	\]
	is a positive $M_n(\Calg_*)$-valued measure, and $\inst(S,-)^{(n)}$ is positive for every $\inst$.
	Since $M_n(\Calg_*)=M_n(\Calg)_*$, \cref{lem:ozawa_positive_integral}\ref{item:ozawa_positive_integral_positive} applies with $M_n(\Calg)$ in place of $\Calg$ and shows that the integral is positive in $M_n(\Calg_*)$.
	Thus $\mu(\instb)(S,-)$ is completely positive.

	To prove that $\mu(\instb)(S,-)$ is trace-nonincreasing, fix $\rho \in \Calg_{*,+}$.
	The vector measure $\instb(-,\rho)$ is positive and has total mass of trace $\tau(\rho)$, since $\instb(\instmon_\Calg X,-)$ is trace-preserving.
	Since every $\inst(X,-)$ is trace-preserving, \cref{lem:ozawa_positive_integral}\ref{item:ozawa_positive_integral_trace} applied to $\instb(-,\rho)$ gives
	\[
		\tau(\mu(\instb)(X,\rho))
		=
		\tau(\rho).
	\]
	Thanks to $\inst(S,-)\leq\inst(X,-)$ for every $\inst$, \cref{lem:ozawa_integral_monotone} gives
	\[
		\mu(\instb)(S,\rho)
		\leq
		\mu(\instb)(X,\rho).
	\]
	Applying the positive functional $\tau$ results in
	\[
		\tau(\mu(\instb)(S,\rho))
		\leq
		\tau(\mu(\instb)(X,\rho))
		=
		\tau(\rho).
	\]
	Hence $\mu(\instb)(S,-)$ is trace-nonincreasing, and the case $S=X$ above shows that $\mu(\instb)(X,-)$ is a channel.
	In particular, we conclude that $\mu(\instb)(S,-)$ is indeed a quantum operation.

	To finish the proof that $\mu(\instb)$ is a quantum instrument, it remains to show countable additivity in $S$.
	Thus let $(S_n)_{n \in \N}$ be pairwise disjoint and put $S\coloneqq\bigcup_n S_n$ and $S^{(N)}\coloneqq\bigcup_{n=1}^N S_n$ for all $N$.
	Then for every $\inst$ and $\rho$, we have
	\[
		\inst(S^{(N)},\rho)
		\longrightarrow
		\inst(S,\rho)
	\]
	in norm.
	Therefore by \Cref{lem:ozawa_integral_convergence},
	\[
		\mu(\instb)(S^{(N)},\rho)
		\longrightarrow
		\mu(\instb)(S,\rho)
	\]
	in norm.
	Since finite additivity of $\mu(\instb)$ is immediate from linearity of the integral, this establishes countable additivity.

	Finally, we prove measurability of $\mu$.
	By definition of the $\sigma$-algebra on $\instmon_\Calg X$,
	it is enough to show that for every $S \in \Sigma_X$ and $\rho \in \Calg_*$, the map
	\begin{align*}
		\instmon_\Calg\instmon_\Calg X &\longrightarrow \Calg_* \\
		\instb & \longmapsto \mu(\instb)(S,\rho)
	\end{align*}
	is measurable, i.e.~that the integral in~\eqref{eq:ozawa_multiplication} depends measurably on $\instb$.
	To this end, use \cref{lem:ozawa_integral_countable_approx}\ref{item:countably_valued_approximation} to choose measurable maps with countable image
	\[
		L_n:\instmon_\Calg X \longrightarrow \Op(\Calg)
	\]
	which converge pointwise to $L : \inst \mapsto \inst(S,-)$ in the strong-operator topology.
	If $L_n(\inst)=\Phi_i$ for $\inst\in R_i$, with $(R_i)_{i \in \N}$ forming a countable measurable partition of $\instmon_\Calg X$, then the map
	\[
		\instb \longmapsto
		\int_{\inst \in \instmon_\Calg X} L_n(\inst) \left(\instb(\mathrm{d}\inst,\rho)\right)
		=
		\sum_i \Phi_i\bigl(\instb(R_i,\rho)\bigr)
	\]
	is measurable by the definition of the $\sigma$-algebra on $\instmon_\Calg\instmon_\Calg X$, since the sum is the pointwise norm limit of its finite partial sums.
	Therefore the desired map is the pointwise norm limit of these measurable maps by \cref{lem:ozawa_integral_convergence}.
	Thus $\mu$ is measurable.
\end{proof}

\begin{thm}
	\label{thm:ozawa_monad}
	Let $\Calg$ be a von Neumann algebra that is a countable direct sum of type I factors with separable predual.
	Then the functor $\instmon_\Calg$, with unit $\eta$ and multiplication $\mu$ defined above, is a monad on $\Meas$.
\end{thm}

\begin{proof}
	By \Cref{prop:ozawa_unit_multiplication}, the unit and multiplication are measurable and have the required codomains.
	It remains to prove naturality and the monad laws.

	Put $R\coloneqq\Op(\Calg)$ and let $\Cont$ denote the continuation monad $\Cont_R$ on $\Meas$ from \Cref{app:continuation_monads}.
	For the rest of the proof, we write $\instmon$ for $\instmon_\Calg$.
	For every measurable space $X$, define
	\begin{align*}
		\theta \: : \: \instmon X &\longrightarrow \Cont\, X \\
		\inst &\longmapsto \left( L\mapsto \left[\rho\mapsto \int_{x\in X} L(x) \left(\inst(\mathrm{d}x,\rho)\right)\right]\right),
	\end{align*}
	where $L:X\to R$ ranges over all measurable maps.
	For fixed $\inst$ and $L$, the value $\theta(\inst)(L)$ is a quantum operation by definition of the quantum operation integral.
	Complete positivity follows by applying \cref{lem:ozawa_positive_integral}\ref{item:ozawa_positive_integral_positive} at matrix levels, and trace-nonincrease follows first for $L$ with countable image from the formula in \cref{lem:ozawa_integral_countable_approx}\ref{item:countably_valued_integral}, and then in general by approximation and \Cref{lem:ozawa_integral_convergence}.

	We show that the map $\theta$ is measurable.
	Since $\Cont\, X$ carries the product $\sigma$-algebra, it is enough to fix $L:X\to R$ and show that $\inst\mapsto\theta(\inst)(L)$ is measurable as a map $\instmon X\to R$.
	Since the measurable structure on $R=\Op(\Calg)$ is generated by evaluations on $\Calg_*$, it is enough to evaluate at any $\rho\in\Calg_*$.
	The resulting map is measurable $L$ with countable image by \cref{lem:ozawa_integral_countable_approx}\ref{item:countably_valued_integral} and the defining evaluation maps on $\instmon X$; the general case follows by approximation and \Cref{lem:ozawa_integral_convergence}.

	We now prove the same claims as in \Cref{thm:classical_state_monad}, with sums replaced by quantum operation integrals.
	First, $\theta$ is natural.
	For $f:X\to Y$, this means that the square
	\[
		\begin{tikzcd}
			\instmon X \arrow[r, "\theta"] \arrow[d, "f_*"'] & \Cont\, X \arrow[d, "\Cont\, f"] \\
			\instmon Y \arrow[r, "\theta"] & \Cont\, Y
		\end{tikzcd}
	\]
	commutes.
	Indeed, for $\inst\in\instmon X$, $L:Y\to R$, and $\rho\in\Calg_*$, the change-of-variables formula gives
	\begin{align*}
		\theta(f_*\inst)(L)(\rho)
		&=
		\int_{y\in Y} L(y) \left((f_*\inst)(\mathrm{d}y,\rho)\right) \\
		&=
		\int_{x\in X} L(f(x)) \left(\inst(\mathrm{d}x,\rho)\right) \\
		&=
		\theta(\inst)(L\circ f)(\rho)
		=
		(\Cont\, f)(\theta(\inst))(L)(\rho).
	\end{align*}

	Next, $\theta$ is injective.
	For $S\in\Sigma_X$, let $L_S:X\to R$ be the measurable map
	\[
		L_S(x)\coloneqq \delta_{x\in S}\id_{\Calg_*}.
	\]
	Then for every $\rho\in\Calg_*$,
	\[
		\theta(\inst)(L_S)(\rho)
		=
		\int_{x\in X} L_S(x) \left(\inst(\mathrm{d}x,\rho)\right)
		=
		\inst(S,\rho).
	\]
	Hence $\theta(\inst)$ determines $\inst(S,\rho)$ for every $S$ and $\rho$, and therefore determines $\inst$.

	The unit compatibility diagram
	\[
		\begin{tikzcd}
			& X \arrow[dl, "\eta"'] \arrow[dr, "\eps"] & \\
			\instmon X \arrow[rr, "\theta"'] && \Cont\,X
		\end{tikzcd}
	\]
	commutes, since for $L:X\to R$ and $\rho\in\Calg_*$,
	\[
		\theta(\eta(x))(L)(\rho)
		=
		\int_{y\in X} L(y) \left(\eta(x)(\mathrm{d}y,\rho)\right)
		=
		L(x) \left(\rho\right)
		=
		\eps(x)(L)(\rho).
	\]
	Here, the second equality is the Dirac case of the quantum operation integral, as in \Cref{lem:quantum_operation_integral_dirac}.

	Finally, the multiplication compatibility diagram
	\[
		\begin{tikzcd}[column sep=large]
			\instmon\instmon X \arrow[r, "\theta"] \arrow[d, "\mu"'] & \Cont\,\instmon X \arrow[r, "\Cont\,\theta"] & \Cont\,\Cont\,X \arrow[d, "\nu"] \\
			\instmon X \arrow[rr, "\theta"'] && \Cont\,X
		\end{tikzcd}
	\]
	commutes.
	To see this, let $\instb\in\instmon\instmon X$ and measurable $L:X\to R$ be given and evaluate at $\rho\in\Calg_*$.
	For the lower-left path, we apply \Cref{lem:quantum_operation_integral_fubini} to the kernel $\kappa_\inst(S,\sigma)\coloneqq\inst(S,\sigma)$ and the vector measure $\mfrak=\instb(-,\rho)$.
	The hypotheses of the lemma then are exactly the quantum instrument axioms, together with the defining evaluation maps on $\instmon X$.
	This gives
	\begin{align*}
		(\theta\circ\mu)(\instb)(L)(\rho)
		&=
		\int_{x\in X} L(x) \left(\mu(\instb)(\mathrm{d}x,\rho)\right) \\
		&=
		\int_{\inst\in\instmon X}
		\left(\int_{x\in X} L(x) \left(\inst(\mathrm{d}x,-)\right)\right)
		\left(\instb(\mathrm{d}\inst,\rho)\right) \\
		&=
		\int_{\inst\in\instmon X}
		\theta(\inst)(L) \left(\instb(\mathrm{d}\inst,\rho)\right).
	\end{align*}
	For the upper-right path, with $\ev_L:\Cont\, X\to R$ denoting evaluation at $L$, we get
	\begin{align*}
		(\nu \circ (\Cont\,\theta) \circ \theta)(\instb)(L)(\rho)
		&=
		(\Cont\,\theta)(\theta(\instb))(\ev_L)(\rho) \\
		&=
		\theta(\instb)(\ev_L \circ \theta)(\rho) \\
		&=
		\int_{\inst\in\instmon X}
		(\ev_L\circ\,\theta)(\inst) \left(\instb(\mathrm{d}\inst,\rho)\right) \\
		&=
		\int_{\inst\in\instmon X}
		\theta(\inst)(L) \left(\instb(\mathrm{d}\inst,\rho)\right).
	\end{align*}

	Overall, we have that $\theta$ is natural, injective, and compatible with unit and multiplication.
	The naturality of $\eta$ and $\mu$ and the monad laws for $\instmon$ therefore follow from the corresponding facts for $\Cont$.
\end{proof}

The following strength is the direct analogue of the usual strength of probability monads, and in particular of the Giry monad, for which $(x,p)$ is sent to the product measure $\delta_x\otimes p$~\cite[Section~1]{FritzPerrone}.
But now, the probability measure is replaced by a quantum instrument, and pairing with the deterministic value $x$ means evaluating the original instrument on the section $S_x$ for a given measurable subset $S\subseteq X\times Y$ as follows.

\begin{prop}
	\label{prop:ozawa_strength}
	With respect to the cartesian monoidal structure on $\Meas$, the monad $\instmon_\Calg$ is strong,
	with strength given by the measurable maps
	\begin{align*}
		\st \: : \: X\times \instmon_\Calg Y & \longrightarrow \instmon_\Calg(X\times Y) \\
		(x,\inst) & \longmapsto \left( (S,\rho) \mapsto \inst(S_x,\rho) \right),
	\end{align*}
	where
	\[
		S_x\coloneqq \{y\in Y\mid (x,y)\in S\}.
	\]
\end{prop}

\begin{proof}
	For fixed $x$ and $\inst$, the assignment $(S,\rho)\mapsto \inst(S_x,\rho)$ is a quantum instrument.
	Indeed, every section $S_x$ of a measurable set in a product space $X \times Y$ is measurable, taking sections preserves complements and countable disjoint unions, and $(X\times Y)_x=Y$.

	We next prove measurability of $\st$.
	For fixed $\rho\in\Calg_*$, let $\mathcal D_\rho$ be the class of all $S\in\Sigma_{X\times Y}$ for which
	\[
		(x,\inst)\longmapsto \inst(S_x,\rho)
	\]
	is measurable as a map $X\times\instmon_\Calg Y\to\Calg_*$.
	This class is a Dynkin system.
	Indeed, complements use $\inst(Y\setminus S_x,\rho)=\inst(Y,\rho)-\inst(S_x,\rho)$, and countable disjoint unions use countable additivity of $\inst$ together with the fact that pointwise norm limits of measurable $\Calg_*$-valued maps are measurable.
	Moreover $\mathcal D_\rho$ contains all measurable rectangles $A\times B$, because
	\[
		\inst((A\times B)_x,\rho)
		=
		\delta_{x\in A}\inst(B,\rho).
	\]
	Hence by the $\pi$-$\lambda$ theorem, we have $\mathcal D_\rho=\Sigma_{X\times Y}$.
	This is exactly the measurability required by the defining $\sigma$-algebra on $\instmon_\Calg(X\times Y)$.

	Naturality in both variables means that, for measurable $f:X\to X'$ and $g:Y\to Y'$, the square
	\[
		\begin{tikzcd}
			X\times\instmon_\Calg Y \arrow[r, "\st"] \arrow[d, "f\times\instmon_\Calg g"'] & \instmon_\Calg(X\times Y) \arrow[d, "\instmon_\Calg(f\times g)"] \\
			X'\times\instmon_\Calg Y' \arrow[r, "\st"'] & \instmon_\Calg(X'\times Y')
		\end{tikzcd}
	\]
	commutes.
	This is immediate from taking sections:
	for measurable $f:X\to X'$ and $g:Y\to Y'$ and $S\in\Sigma_{X'\times Y'}$, we get
	\[
		((f\times g)^{-1}S)_x
		=
		g^{-1}(S_{f(x)}),
	\]
	from which the required naturality condition follows straightforwardly.

	The monoidal unit coherence diagram is
	\[
		\begin{tikzcd}
			1\times\instmon_\Calg Y \arrow[rr, "\st"] \arrow[dr, "\lambda"'] && \instmon_\Calg(1\times Y) \arrow[dl, "\instmon_\Calg\lambda"] \\
			& \instmon_\Calg Y
		\end{tikzcd}
	\]
	where $\lambda$ denotes the canonical left unitor.
	This diagram commutes because, under the canonical identifications $1\times Y\cong Y$, the unit coherence says that $\st(*,\inst)$ is identified with $\inst$, which follows from $S_*=S$.
	The associativity coherence diagram is
	\[
		\begin{tikzcd}[column sep=large]
			(X\times Y)\times\instmon_\Calg Z \arrow[r, "\st"] \arrow[d, "\alpha"'] & \instmon_\Calg((X\times Y)\times Z) \arrow[rr, "\instmon_\Calg\alpha"] && \instmon_\Calg(X\times(Y\times Z)) \\
			X\times(Y\times\instmon_\Calg Z) \arrow[r, "\id_X\times\st"'] & X\times\instmon_\Calg(Y\times Z) \arrow[rr, "\st"'] && \instmon_\Calg(X\times(Y\times Z)) \arrow[u, equal]
		\end{tikzcd}
	\]
	where $\alpha$ denotes the canonical associator.
	Both paths send $(x,y,\inst)$ to the instrument
	\[
		(S,\rho)\longmapsto \inst(S_{(x,y)},\rho).
	\]

	It remains to know that this strength of the underlying functor is compatible with the monad unit and multiplication.
	This is implied by a result of McDermott and Uustalu~\cite[Proposition~4.7]{McDermottUustalu}, which shows that the compatibility with the unit and multiplication hold automatically, provided that the cartesian action of $\Meas$ on itself is well-pointed in their sense. 
	This means that measurable maps $W \times X\to Y$ are determined by their restrictions along the points $w:1\to W$, which is indeed the case.
\end{proof}

\appendix

\section{Continuation monads on \texorpdfstring{$\Set$}{Set} and \texorpdfstring{$\Meas$}{Meas}}
\label{app:continuation_monads}

We recall two closely related continuation monads, one on $\Set$ and one on $\Meas$.
We start with the case of $\Meas$, since it makes a clearer distinction between the two categories that are involved in the construction.
Fix some $R \in \Meas$, thought of as a set of return values.
Then the universal property of products (or powers) in $\Meas$ gives bijections
\begin{equation}
	\label{eq:continuation_adjunction}
	\Meas(X,R^A)
	\,\cong\,
	\Set(A,\Meas(X,R))
\end{equation}
natural in $X \in \Meas$ and $A \in \Set$.
Thus we are dealing with a contravariant adjunction between $\Meas$ and $\Set$.
Its induced monad on $\Meas$ is the continuation monad, which has underlying endofunctor
\[
	\Cont_R X
	\coloneqq
	R^{\Meas(X,R)},
\]
whose action on morphisms $f:X\to Y$ is given by
\[
	\Cont_R(f)(T)(L)
	=
	T(L\circ f)
\]
for all $T\in\Cont_R X$ and continuations $L \in \Meas(Y,R)$.
Both of these are easily derived from~\eqref{eq:continuation_adjunction} by the standard construction of a monad from an adjunction.
The unit has components
\begin{align*}
	\eps \: : \: X & \longrightarrow \Cont_R X \\
	x & \longmapsto (L \mapsto L(x)),
\end{align*}
and the multiplication is given by
\begin{align*}
	\nu \: : \: \Cont_R\Cont_R X & \longrightarrow \Cont_R X \\
	\mathcal T & \longmapsto (L \mapsto \mathcal T(\ev_L)),
\end{align*}
where $\ev_L:\Cont_R X\to R$ is the coordinate projection to the continuation $L$.

If $R$ is a mere set, then the same construction gives a continuation monad on $\Set$ from the contravariant adjunction
\[
	\Set(X,R^A)
	\,\cong\,
	\Set(A,\Set(X,R))
\]
with underlying endofunctor
\[
	\Cont_R X
	\coloneqq
	R^{\Set(X,R)},
\]
and with unit and multiplication given by the same formulas as above, but also with $\Meas$ replaced by $\Set$.

\section{The quantum operation integral}
\label{app:quantum_operation_integral}

In this appendix, we develop the integral of quantum operations against state-valued vector measures.
This is the notion of integral that appears in the definition of the multiplication of the quantum instrument monad in \Cref{sec:ozawa_monad}.
Throughout, $X$ and $Y$ are measurable spaces and
\[
	\mfrak : \Sigma_X \longrightarrow \Calg_*
\]
is a countably additive vector measure of finite variation.
Moreover, we assume that $\Calg$ is a von Neumann algebra such that the predual $\Calg_*$ is separable and has the Radon--Nikodym property (\cref{rem:radon_nikodym_property}).
We use the standard Radon--Nikodym theorem for Banach-space-valued vector measures~\cite{DiestelUhl}.
This shows that there is a measurable function
\[
	D_\mfrak : X \longrightarrow \Calg_*
\]
which is Bochner integrable with respect to the variation measure $|\mfrak|$ and such that
\begin{equation}
	\label{vector_measure_radon_nikodym}
	\mfrak(R) = \int_{x \in R} D_\mfrak(x)\,|\mfrak|(\mathrm{d}x)
	\qquad
	\forall R \in \Sigma_X.
\end{equation}

\begin{defn}
	If $L : X \to \Op(\Calg)$ is measurable, then its \newterm{quantum operation integral} against $\mfrak$ is
	\[
		\int_{x \in X} L(x) \left(\mfrak(\mathrm{d}x)\right)
		\coloneqq
		\int_{x \in X} L(x)\bigl(D_\mfrak(x)\bigr)\,|\mfrak|(\mathrm{d}x),
	\]
	where the integral on the right is a Bochner integral in $\Calg_*$.
\end{defn}

The notation on the left is merely suggestive: $\mfrak(\mathrm{d}x)$ is not a pointwise value, but indicates that the quantum operation $L(x)$ acts on the vector-measure increment, as made precise by the Radon--Nikodym formula on the right.
The integrand is Bochner measurable by \Cref{lem:operation_evaluation_measurable}, and it is dominated in norm by $\|D_\mfrak(x)\|$, since quantum operations are contractions.
Thus the integral is well-defined.
More generally, one may use any finite positive measure $\nu$ dominating $|\mfrak|$ and the corresponding density $D_\nu$ of $\mfrak$ with respect to $\nu$.
This gives the same value and makes linearity in the vector measure variable immediate by passing to a common dominating measure.
Linearity in the integrand $L$ is also immediate from linearity of the Bochner integral.

A simple special case of the quantum operation integral is when the vector measure is a Dirac measure:

\begin{lem}
	\label{lem:quantum_operation_integral_dirac}
	Let $x\in X$ and $\rho\in\Calg_*$, and consider the Dirac vector measure
	\[
		\delta_x^\rho(S)
		\coloneqq
		\delta_{x\in S}\rho.
	\]
	Then for every measurable $L:X\to\Op(\Calg)$, we have
	\[
		\int_{y\in X} L(y)\left(\delta_x^\rho(\mathrm{d}y)\right)
		=
		L(x)(\rho).
	\]
\end{lem}

\begin{proof}
	If $\rho=0$, this is immediate.
	Otherwise $|\delta_x^\rho|=\|\rho\|\delta_x$, and the Radon--Nikodym density is equal to $\rho/\|\rho\|$ at $x$.
	Hence the definition of the integral gives
	\[
		\int_{y\in X} L(y)\left(\delta_x^\rho(\mathrm{d}y)\right)
		=
		\int_{y\in X} L(y)\left(D_{\delta_x^\rho}(y)\right)\,|\delta_x^\rho|(\mathrm{d}y)
		=
		L(x)(\rho).
		\qedhere
	\]
\end{proof}

\begin{rem}
	\label{rem:quantum_operation_integral_discussion}
	We make some comments about the quantum operation integral in general.
	\begin{enumerate}
		\item\label{item:quantum_operation_integral_discussion_positive}
			If $\mfrak$ is positive, then the variation measure is simply
			\[
				|\mfrak|=\tau\circ\mfrak,
			\]
			since $\|\rho\|=\tau(\rho)$ for positive $\rho\in\Calg_*$.
			However, positivity alone does not imply the existence of a Bochner density $D_\mfrak:X\to\Calg_*$,
			and in particular the Radon--Nikodym hypothesis cannot be dropped merely by assuming that $\mfrak$ is positive.

			This is illustrated by the following standard example~\cite[Example~III.1.2]{DiestelUhl}.
			Take $\Calg=L^\infty([0,1])$, so that $\Calg_*=L^1([0,1])$, and define the positive vector measure on $X = [0,1]$ given by
			\[
				\mfrak(E)\coloneqq\chi_E.
			\]
			Then $|\mfrak|$ is the Lebesgue measure.
			If a Bochner density $D_\mfrak$ existed, then for every $g\in L^\infty([0,1])$ we would have for every $E \in \Sigma_X$ that
			\[
				\int_{x \in E} \langle g, D_\mfrak(x)\rangle\,\mathrm{d}x
				=
				\left\langle g, \int_{x \in E} D_\mfrak(x)\,\mathrm{d}x \right\rangle
				= 
				\langle g, \mfrak(E)\rangle
				=
				\int_{t \in E} g(t)\,\mathrm{d}t,
			\]
			where $\langle -, - \rangle$ denotes the dual pairing between $L^\infty([0,1])$ and $L^1([0,1])$.
			This would imply $\langle D_\mfrak(x),g\rangle=g(x)$ for almost every $x$, which would make $D_\mfrak(x)$ represent point evaluation at $x$.
			This is impossible because point evaluation is not a normal functional on $L^\infty([0,1])$.
			Thus positivity gives scalar Radon--Nikodym derivatives after testing against observables, but these need not assemble into a $\Calg_*$-valued density.
		\item
			Cycon and Hellwig~\cite[Appendix]{CyconHellwig1977} define a related notion of integral in the setting of generalized probabilistic theories,
			where a generalized probabilistic theory is given by a base norm space of states $V$ with dual effect space $V'$.
			They integrate bounded measurable\footnote{Their actual measurability condition is a bit more permissive and involves $\alpha$, similar to how Lebesgue measurability is a bit more permissive than Borel measurability but is relative to the Lebesgue measure.}
			families of states $\varphi : X \to V$ against a POVM $\alpha:\Sigma_X\to V'$, resulting in
			\begin{align*}
				\int_{x\in X}\varphi(x)\,\langle\alpha(\mathrm{d}x),-\rangle \: : \: V \longrightarrow V.
			\end{align*}
			For a simple function $\varphi=\sum_i v_i\chi_{E_i}$, this integral is defined by
			\[
				\int_{x\in X}\varphi(x)\,\langle\alpha(\mathrm{d}x),-\rangle
				\coloneqq
				\sum_i v_i\,\langle\alpha(E_i),-\rangle,
			\]
			and the general case is obtained by approximation as for the Lebesgue integral.

			One can specialize the Cycon--Hellwig integral to the quantum case by taking $V = \Calg_*$.
			Then while the Cycon--Hellwig integral integrates a family of states against a POVM,
			our quantum operations integral integrates a family of quantum operations against a state-valued measure.
			This seems necessary because we have not been able to find a definition of the monad multiplication~\eqref{eq:ozawa_multiplication} in terms of the Cycon--Hellwig integral.
			\item
				One may ask whether the formula for countably valued integrands in \cref{lem:ozawa_integral_countable_approx}\ref{item:countably_valued_integral} could instead be taken as the definition of the integral, by approximating a general $L$ by countably valued $L_n$.
				This does not seem to work without the Radon--Nikodym density, because the resulting limit need not be independent of the approximating sequence.
				For example, take again $\Calg=L^\infty([0,1])$, $\Calg_*=L^1([0,1])$, and $\mfrak(E)=\chi_E$ as above.
				For every $n$, partition $[0,1]$ into measurable sets $I_{n,1},\ldots,I_{n,n}$ of measure $1/n$.
				Define $L_n:[0,1]\to\Op(\Calg)$ by declaring that for every $k$ and $x\in I_{n,k}$,
				\[
					L_n(x)(\rho)
					\coloneqq
					n \left(\int_{t\in I_{n,k}}\rho(t)\,\mathrm{d}t\right) \chi_{I_{n,k}}.
				\]
				This is a positive contraction, hence a quantum operation.
				Then $L_n(x)(\rho)\to 0$ in $\Calg_*$ for every fixed $x$ and $\rho \in \Calg_*$, so $L_n$ converges pointwise strongly to the zero integrand.
				But the countably valued integral formula would give
				\[
					\int_{x\in[0,1]} L_n(x)\left(\mfrak(\mathrm{d}x)\right)
					=
					\sum_{k=1}^n L_n(x_{n,k})(\chi_{I_{n,k}})
					=
					\sum_k\chi_{I_{n,k}}
					=
					\chi_{[0,1]},
				\]
				where $x_{n,k}\in I_{n,k}$.
				Thus approximation by countably valued integrands would assign a nonzero limiting value to the zero integrand.

				The Radon--Nikodym density in our definition rules out precisely this kind of pathology.
		\end{enumerate}
	\end{rem}

We now investigate some basic properties of our quantum operation integral, still assuming that $\Calg_*$ is separable and has the Radon--Nikodym property.

\begin{lem}
	\label{lem:ozawa_integral_countable_approx}
	Let $L : X \to \Op(\Calg)$ be measurable.
	\begin{enumerate}
		\item\label{item:countably_valued_integral}
			If $L$ is countably valued, say $L(x)=\Phi_i$ for $x\in S_i$, where $(S_i)_{i\in\N}$ is a countable measurable partition of $X$, then
			\[
				\int_{x \in X} L(x) \left(\mfrak(\mathrm{d}x)\right)
				=
				\sum_{i\in\N}\Phi_i(\mfrak(S_i)),
			\]
			with norm convergence of the sum.
		\item\label{item:countably_valued_approximation}
			There is a sequence of measurable maps $L_n:X\to\Op(\Calg)$ with countable image and such that
			\[
				L_n(x)(\rho)\longrightarrow L(x)(\rho)
			\]
			in norm for every $x\in X$ and $\rho\in\Calg_*$.
			If $L$ takes values in channels, then the $L_n$ can be chosen to take values in channels as well.
	\end{enumerate}
\end{lem}

\begin{proof}
	For the first claim, we apply the definition of the integral and decompose the Bochner integral over the partition $(S_i)_i$:
	\begin{align*}
		\int_{x \in X} L(x) \left(\mfrak(\mathrm{d}x)\right)
		&=
		\sum_i \int_{x \in S_i} \Phi_i\bigl(D_\mfrak(x)\bigr)\,|\mfrak|(\mathrm{d}x) \\
		&=
		\sum_i \Phi_i\left(\int_{x \in S_i} D_\mfrak(x)\,|\mfrak|(\mathrm{d}x)\right) \\
		&=
		\sum_i \Phi_i(\mfrak(S_i)).
	\end{align*}
	The sum converges absolutely in norm because
	\[
		\sum_i \|\Phi_i(\mfrak(S_i))\|
		\leq
		\sum_i \|\mfrak(S_i)\|
		\leq
		|\mfrak|(X).
	\]

	For the second claim, choose a countable dense subset $(\rho_k)_{k\in\N}$ of $\Calg_*$ and consider the metric
	\[
		d(\Phi,\Psi)
		\coloneqq
		\sum_{k=1}^{\infty}2^{-k}\min\{1,\|\Phi(\rho_k)-\Psi(\rho_k)\|\}
	\]
	on $\Op(\Calg)$.
	This metric generates the measurable structure on $\Op(\Calg)$, and convergence in $d$ implies strong operator convergence because all quantum operations are contractions.
	Moreover, $\Op(\Calg)$ is separable for this metric, since it embeds into the countable product $(\Calg_*)^\N$ of separable metric spaces by \Cref{rem:operation_space_standard_borel}.
	Choose a countable dense subset $(\Phi_i)_{i\in\N}$ and define $L_n$ by sending $x$ to the first $\Phi_i$ with $d(\Phi_i,L(x))<1/n$.
	This is the standard construction of countably valued approximants for measurable maps into separable metric spaces~\cite[Section~4]{AliprantisBorder}.
	It gives $d(L_n(x),L(x))\to 0$, hence the desired strong-operator pointwise convergence.

	If $L$ takes values in channels, use instead a countable dense subset of the channel subspace.
	This subspace is closed under pointwise norm limits: complete positivity, trace nonincrease, and the trace-preserving equality are all closed conditions after evaluating on positive elements of the relevant matrix levels.
	Hence the same construction gives channel-valued approximants.
\end{proof}

\begin{lem}
	\label{lem:ozawa_integral_convergence}
	Let $L_n,L:X\to\Op(\Calg)$ be measurable maps such that
	\begin{equation}
		\label{pointwise_convergence_integrand}
		L_n(x)(\rho) \longrightarrow L(x)(\rho)
	\end{equation}
	in norm for every $x\in X$ and $\rho\in\Calg_*$.
	Then
	\[
		\int_{x \in X} L_n(x) \left(\mfrak(\mathrm{d}x)\right)
		\longrightarrow
		\int_{x \in X} L(x) \left(\mfrak(\mathrm{d}x)\right)
	\]
	in norm.
\end{lem}

\begin{proof}
	By definition of the integral, the integrands are
	\[
		x\longmapsto L_n(x)\bigl(D_\mfrak(x)\bigr)
		\qquad\text{and}\qquad
		x\longmapsto L(x)\bigl(D_\mfrak(x)\bigr).
	\]
	They converge pointwise in norm by the assumption.
	Moreover,
	\[
		\left\|L_n(x)\bigl(D_\mfrak(x)\bigr)-L(x)\bigl(D_\mfrak(x)\bigr)\right\|
		\leq
		2\|D_\mfrak(x)\|,
	\]
	and the right-hand side is integrable with respect to $|\mfrak|$.
	The claim follows from the Bochner dominated convergence theorem.
\end{proof}

\begin{lem}
	\label{lem:ozawa_positive_integral}
	Let $\mfrak:\Sigma_X\to\Calg_*$ be a positive measure of finite variation and let $L:X\to\Op(\Calg)$ be measurable.
	\begin{enumerate}
		\item\label{item:ozawa_positive_integral_positive}
			The integral
			\[
				\int_{x \in X} L(x) \left(\mfrak(\mathrm{d}x)\right)
			\]
			is positive in $\Calg_*$.
		\item\label{item:ozawa_positive_integral_trace}
			If every $L(x)$ is trace-preserving, then
			\[
				\tau\left(\int_{x \in X} L(x) \left(\mfrak(\mathrm{d}x)\right)\right)
				=
				\tau(\mfrak(X)).
			\]
	\end{enumerate}
\end{lem}

\begin{proof}
	For countably valued $L$, both claims follow from \cref{lem:ozawa_integral_countable_approx}\ref{item:countably_valued_integral}.
	Indeed, if $L(x)=\Phi_i$ on the measurable partition $S_i$, then
	\[
		\int_{x \in X} L(x) \left(\mfrak(\mathrm{d}x)\right)
		=
		\sum_i \Phi_i(\mfrak(S_i)).
	\]
	Every term is positive, proving positivity.
	If all $\Phi_i$ are trace-preserving, then
	\[
		\tau\left(\sum_i\Phi_i(\mfrak(S_i))\right)
		=
		\sum_i\tau(\mfrak(S_i))
		=
		\tau(\mfrak(X)).
	\]

	For general $L$, use \cref{lem:ozawa_integral_countable_approx}\ref{item:countably_valued_approximation} and then pass to the limit with \Cref{lem:ozawa_integral_convergence}.
	The positive cone in $\Calg_*$ is norm closed, and the trace functional $\tau$ is norm continuous.
\end{proof}

\begin{lem}
	\label{lem:ozawa_integral_monotone}
	Let $\mfrak:\Sigma_X\to\Calg_*$ be a positive measure of finite variation and $K,L:X\to\Op(\Calg)$ measurable maps such that $K(x)\leq L(x)$ in the cp order for every $x\in X$.
	Then
	\[
		\int_{x \in X} K(x) \left(\mfrak(\mathrm{d}x)\right)
		\leq
		\int_{x \in X} L(x) \left(\mfrak(\mathrm{d}x)\right).
	\]
\end{lem}

\begin{proof}
	The pointwise difference $L-K$ is a measurable map $X\to\Op(\Calg)$.
	By \cref{lem:ozawa_positive_integral}\ref{item:ozawa_positive_integral_positive},
	\[
		\int_{x \in X} (L(x)-K(x)) \left(\mfrak(\mathrm{d}x)\right)
	\]
	is positive.
	Linearity of the integral in the integrand gives the claim.
\end{proof}

We next prove a change of variables formula.

\begin{lem}
	\label{lem:ozawa_change_variables}
	Let $f:X\to Y$ be measurable, and let $L:Y\to\Op(\Calg)$ be measurable.
	Then
	\[
		\int_{y \in Y} L(y) \left((f_*\mfrak)(\mathrm{d}y)\right)
		=
		\int_{x \in X} L(f(x)) \left(\mfrak(\mathrm{d}x)\right).
	\]
\end{lem}

\begin{proof}
	For countably valued $L$, say $L(y)=\Phi_i$ on the measurable partition $T_i$ of $Y$, the claim follows directly from \cref{lem:ozawa_integral_countable_approx}\ref{item:countably_valued_integral}:
	\[
		\int_{y \in Y} L(y) \left((f_*\mfrak)(\mathrm{d}y)\right)
		=
		\sum_i\Phi_i(\mfrak(f^{-1}(T_i)))
		=
		\int_{x \in X} L(f(x)) \left(\mfrak(\mathrm{d}x)\right).
	\]
	For general $L$, choose countably valued approximants $L_n$ as in \cref{lem:ozawa_integral_countable_approx}\ref{item:countably_valued_approximation}.
	Then $L_n\circ f$ approximates $L\circ f$ pointwise in the strong-operator topology, and the result follows by applying \Cref{lem:ozawa_integral_convergence} on both sides.
\end{proof}

Our final lemma is a Fubini-type result for quantum operation integrals, which is relevant for the proof of associativity of the multiplication of the quantum instrument monad in \Cref{sec:ozawa_monad}.

\begin{lem}
	\label{lem:quantum_operation_integral_fubini}
	Let $X$ and $Y$ be measurable spaces, and let
	\[
		\kappa_y : \Sigma_X \times \Calg_* \longrightarrow \Calg_*,
		\qquad y\in Y,
	\]
	be a measurable family of quantum-operation-valued measures in the following sense:
	\begin{enumerate}
		\item for every $y\in Y$ and $S\in\Sigma_X$, the map $\kappa_y(S,-)$ is a quantum operation;
		\item for every $y\in Y$ and $\rho\in\Calg_*$, the map $S\mapsto\kappa_y(S,\rho)$ is a finite-variation vector measure;
		\item for every $S\in\Sigma_X$, the map $y\mapsto\kappa_y(S,-)$ is measurable as a map $Y\to\Op(\Calg)$.
	\end{enumerate}
	Let $\mfrak:\Sigma_Y\to\Calg_*$ be a finite-variation vector measure, and define
	\[
		(\kappa_*\mfrak)(S)
		\coloneqq
		\int_{y \in Y} \kappa_y(S,-) \left(\mfrak(\mathrm{d}y)\right)
	\]
	for $S\in\Sigma_X$.
	Then $\kappa_*\mfrak$ is a finite-variation vector measure on $X$.
	Moreover, for every measurable $L:X\to\Op(\Calg)$, the formula
	\[
		\overline L(y)(\rho)
		\coloneqq
		\int_{x \in X} L(x) \left(\kappa_y(\mathrm{d}x,\rho)\right)
	\]
	defines a measurable map $\overline L:Y\to\Op(\Calg)$, and
	\[
		\int_{y \in Y} \overline L(y) \left(\mfrak(\mathrm{d}y)\right)
		=
		\int_{x \in X} L(x) \left((\kappa_*\mfrak)(\mathrm{d}x)\right).
	\]
\end{lem}

\begin{proof}
	We first show that $\kappa_*\mfrak$ is a finite-variation vector measure.
	Finite additivity follows from finite additivity of $S\mapsto\kappa_y(S,\rho)$ and linearity of the quantum operation integral in the integrand.
	If the $(S_n)_{n\in\N}$ are pairwise disjoint, put $S=\bigcup_n S_n$ and $S^{(N)}=\bigcup_{n=1}^N S_n$.
	Then $\kappa_y(S^{(N)},\rho)\to\kappa_y(S,\rho)$ in norm for all $y$ and $\rho$, so \Cref{lem:ozawa_integral_convergence} gives
	\[
		(\kappa_*\mfrak)(S^{(N)})
		\longrightarrow
		(\kappa_*\mfrak)(S).
	\]
	This proves countable additivity.

	For finite variation, write $D_\mfrak$ for the Radon--Nikodym density of $\mfrak$ with respect to $|\mfrak|$.
	If $\rho\in\Calg_{*,+}$ and $(S_i)_i$ is a finite measurable partition of $X$, then
	\[
		\sum_i \|\kappa_y(S_i,\rho)\|
		=
		\sum_i \tau(\kappa_y(S_i,\rho))
		=
		\tau(\kappa_y(X,\rho))
		\leq
		\tau(\rho)
		=
		\|\rho\|.
	\]
	Writing an arbitrary normal functional as a linear combination of four positive normal functionals with total norm at most $4\|\rho\|$ gives
	\[
		\sum_i \|\kappa_y(S_i,\rho)\|
		\leq
		4\|\rho\|
		\qquad
		\forall \rho\in\Calg_*.
	\]
	Hence
	\begin{align*}
		\sum_i \|(\kappa_*\mfrak)(S_i)\|
		&\leq
		\int_{y \in Y}
		\sum_i \|\kappa_y(S_i,D_\mfrak(y))\|
		\,|\mfrak|(\mathrm{d}y) \\
		&\leq
		4\int_{y \in Y}\|D_\mfrak(y)\|\,|\mfrak|(\mathrm{d}y)
		<
		\infty.
	\end{align*}
	Taking the supremum over finite partitions shows that $\kappa_*\mfrak$ has finite variation.

	We turn to the claims involving $\overline L$.
	First suppose that $L$ is countably valued with $L(x)=\Phi_i$ for $x\in S_i$, where $(S_i)_{i\in\N}$ is a measurable partition of $X$.
	Then
	\[
		\overline L(y)(\rho)
		=
		\sum_i \Phi_i(\kappa_y(S_i,\rho)).
	\]
	This is a quantum operation in $\rho$: linearity follows from linearity of the integral in the measure variable,
	complete positivity follows from complete positivity of the summands and norm closedness of the positive cone, while for positive $\rho$,
	\[
		\tau(\overline L(y)(\rho))
		\leq
		\sum_i \tau(\kappa_y(S_i,\rho))
		=
		\tau(\kappa_y(X,\rho))
		\leq
		\tau(\rho).
	\]
	Moreover $y\mapsto\overline L(y)$ is measurable, since for every $\rho$ the map
	\[
		y\longmapsto
		\sum_i \Phi_i(\kappa_y(S_i,\rho))
	\]
	is the pointwise norm limit of measurable finite partial sums.
	For finite-valued $L$, using \cref{lem:ozawa_integral_countable_approx}\ref{item:countably_valued_integral}, finite linearity of the integral in the integrand, and commutation of Bochner integration with bounded linear maps gives
	\begin{align*}
		\int_{y \in Y} \overline L(y) \left(\mfrak(\mathrm{d}y)\right)
		&=
		\int_{y \in Y}
		\left(
			\sum_i \Phi_i\circ\kappa_y(S_i,-)
		\right) \left(\mfrak(\mathrm{d}y)\right) \\
		&=
		\sum_i
		\int_{y \in Y} (\Phi_i\circ\kappa_y(S_i,-)) \left(\mfrak(\mathrm{d}y)\right) \\
		&=
		\sum_i \Phi_i\left(
			\int_{y \in Y} \kappa_y(S_i,-) \left(\mfrak(\mathrm{d}y)\right)
		\right) \\
		&=
		\sum_i \Phi_i((\kappa_*\mfrak)(S_i)) \\
		&=
		\int_{x \in X} L(x) \left((\kappa_*\mfrak)(\mathrm{d}x)\right).
	\end{align*}
	The same identity follows for countably valued $L$ by applying the finite-valued case to the finite partial sums and then passing to the limit by \Cref{lem:ozawa_integral_convergence}.

	For general $L$, choose countably valued $L_n:X\to\Op(\Calg)$ as in \cref{lem:ozawa_integral_countable_approx}\ref{item:countably_valued_approximation}, and let $\overline L_n$ denote the corresponding maps $Y\to\Op(\Calg)$.
	For every $y\in Y$ and $\rho\in\Calg_*$, \cref{lem:ozawa_integral_convergence} applied to the vector measure $R\mapsto\kappa_y(R,\rho)$ gives
	\[
		\overline L_n(y)(\rho)\longrightarrow \overline L(y)(\rho).
	\]
	Since $\Op(\Calg)$ is closed under pointwise norm limits, each $\overline L(y)$ is a quantum operation.
	It follows that $\overline L$ is measurable, because the measurable structure on $\Op(\Calg)$ is generated by evaluations on $\Calg_*$.
	The desired identity follows by applying \cref{lem:ozawa_integral_convergence} to both sides of the identity already proved for $L_n$, once with respect to $\kappa_*\mfrak$ and once with respect to $\mfrak$.
\end{proof}

\newpage
\printbibliography

\end{document}